\documentclass[a4paper, 11pt]{article}
\usepackage{amsthm}
\usepackage[T1]{fontenc}
\usepackage{amsmath}
\usepackage{graphicx}
\usepackage{amsfonts}
\usepackage{amssymb}

\newtheorem{theorem}{Theorem}[section]

\newtheorem{proposition}[theorem]{Proposition}
\newtheorem{corollary}[theorem]{Corollary}
\theoremstyle{definition}

\newtheorem{remark}[theorem]{Remark}
\newtheorem{example}[theorem]{Example}

\begin{document}

\date{2018-4-10}
\title{Fermionic quantum detailed balance and entanglement}
\author{Rocco Duvenhage\\Department of Physics, University of Pretoria, Pretoria 0002 \\South Africa \\rocco.duvenhage@up.ac.za}
\maketitle

\begin{abstract}
A definition of detailed balance tailored to a system of indistinguishable
fermions is suggested and studied using an entangled fermionic state. This is
done in analogy to a known characterization of standard quantum detailed
balance with respect to a reversing operation.

\end{abstract}



\section{Introduction}

In this paper we take the first steps to formulate quantum detailed balance
tailored for a system consisting of indistinguishable fermions, by using an entangled fermionic state.

Classically, detailed balance of a Markov chain and a given probability
distribution means that the probability for the system to make a transition
from one pure state to another is equal to the probability for the opposite
transition. More precisely, $p_{j}\gamma _{jk}=p_{k}\gamma _{kj}$, where the  
$p_{j}$ form the probability distribution over the pure states, 
and the $\gamma _{jk}$ are the transition probabilities. 
In the quantum case, however, one can express detailed balance in
terms of an entangled state of two copies of the system in question, without
direct reference to transition probabilities. See \cite{DS, FR10, FR}. Also
see the closely related papers \cite{BQ, BQ2}. This formulation creates the
possibility to generalize or adapt detailed balance in natural ways which are
not apparent from the classical formulation. See for example \cite{DS2}.

The references mentioned above build on previous work on quantum detailed
balance, where connections with entanglement were, however, not made explicit,
in particular \cite{Ag, Al, CWA, KFGV, M, M83}.

Entanglement of indistinguishable particles is not considered in the 
references above. Since many relevant systems consist of
indistinguishable particles, it is natural to ask whether the formulation of
detailed balance in terms of an entangled state can be successfully adapted to
entangled states of indistinguishable particles.

In this paper we show that this is in fact possible, at least for simple fermionic
systems. We thus obtain a purely quantum mechanical formulation of detailed balance for a system of indistinguishable fermions. 

We focus on obtaining a fermionic version of the so-called 
standard quantum detailed balance condition with respect to a reversing 
operation \cite{FR, FU}. A standard choice of reversing operation will 
be used in the latter, namely transposition of matrices. Our fermionic 
detailed balance condition is then modelled on this situation, and will 
be called \emph{fermionic standard quantum detailed balance}.

We take our fermionic system to be a finite sublattice in a fermion 
lattice. We then consider a second copy of this system elsewhere in the lattice, 
and set up an entangled state for these two systems. In the entangled state 
formulation of detailed balance for the generic case without reference to 
indistinguishable particles, the tensor product of two copies of the system 
being studied, plays a central role. In the fermionic case, however, 
this simple tensor product structure is lost, but its place is filled by 
the lattice structure.

Entangled states for indistinguishable particles are discussed in some detail (both conceptually and technically) in for example \cite[Section 5-4]{P}, \cite{ESBL} and \cite{GMW}. We only need a very particular case however, for which we give a self-contained, although brief, discussion.

Duality of dynamical maps play an important role in quantum detailed balance 
(see for example \cite{FR}). A preliminary investigation into an
analogous duality in the fermionic case is presented in this paper.
However, the analogy with duality in the usual case is only partial. 
In particular, positivity properties in the usual case are not 
present in the fermionic case, possibly indicating that different forms of 
duality should be explored further. 

Tailoring quantum detailed balance to fermionic systems is in line with the
general theme of extending or adapting to fermionic systems, various concepts
and results from either classical probability, or from quantum probabilistic
theories that do not take fermionic behaviour into account. Examples of this
avenue of research can be found in \cite{AFM, AM3, AM, AM2, CF, Fid2, KZE}.

We review the essentials of the fermion lattice in Section
\ref{AfdFermi}. In Section \ref{AfdProdDiag} we treat the fermionic 
entangled state which we need. 
Then, in Section \ref{AfdFb}, we turn to the
formulation of fermionic standard quantum detailed balance in terms of this 
entangled state. The usual case (without reference to indistinguishable particles) of standard quantum detailed balance with respect to a reversing operation is also briefly discussed as part of that section, in order to clarify the origins of the fermionic formulation. 
Section \ref{AfdVb} proceeds with a simple example to illustrate fermionic standard quantum detailed balance. 
A related example is discussed in Section \ref{AfdTeenVb} to show that 
fermionic standard quantum detailed balance indeed differs from the usual 
standard quantum detailed balance condition applied to a fermionic system. 
Duality of dynamical maps is explored in Section \ref{AfdDuaal}. 
The paper concludes in Section \ref{Vrae} with some questions for
possible further research.

\section{The fermion lattice}\label{AfdFermi}


Here we briefly review the framework that we will use, and also set up much of
the notation for the rest of the paper. Our main references are \cite[Section
5.2]{BR2} and \cite{AM}.

Let $h$ denote the Hilbert space for a single fermion and consider the Fock
space
\[
F(h)=\bigoplus_{n\geq0}h^{n}
\]
where
\[
h^{0}:=\mathbb{C}\Psi
\]
with $\Psi$ the vacuum vector, which we can simply take to be $\Psi
=1\in\mathbb{C}$, while for $n=1,2,3,...$
\[
h^{n}:=h\otimes...\otimes h
\]
where $h$ appears $n$ times.

Consider the projection
\[
P:F(h)\rightarrow F(h)
\]
defined by
\[
P\Psi=\Psi
\]
and
\begin{equation}
P(x_{1}\otimes...\otimes x_{n})=\frac{1}{n!}\sum_{\pi\in S_{n}}\varepsilon
_{\pi}x_{\pi(1)}\otimes...\otimes x_{\pi(n)} \label{P}
\end{equation}
for $x_{1},...,x_{n}\in h$, where $n=1,2,3,...$. Here $\varepsilon_{\pi}$
denotes the sign of the permutation $\pi$. We then define the Fermi-Fock space
as
\[
H=PF(h)
\]
and denote its inner product, inherited from $F(h)$, by 
$\langle\cdot,\cdot\rangle$.

We can define creation operators $a^{\ast}(x)$ on $H$ for all $x\in h$ as
follows: First define a creation operator $b^{\ast}(x)$ on a dense subspace of
$F(h)$ by
\[
b^{\ast}(x)\Psi=x
\]
and
\[
b^{\ast}(x)(x_{1}\otimes...\otimes x_{n})=(n+1)^{1/2}x\otimes x_{1}
\otimes...\otimes x_{n}
\]
for all $x,x_{1},...,x_{n}\in h$. The corresponding annihilation operator
$b(x)$, is given by
\[
b(x)\Psi=0
\]
and
\[
b(x)(x_{0}\otimes x_{1}\otimes...\otimes x_{n})=(n+1)^{1/2}\left\langle
x,x_{0}\right\rangle x_{1}\otimes...\otimes x_{n}
\]
for all $x,x_{0},x_{1},...,x_{n}\in h$, where the inner product $\left\langle
\cdot,\cdot\right\rangle $ of $h$ is taken to be linear in the second slot.
The fermionic creation and annihilation operators are then defined as
\[
a^{\ast}(x)=Pb^{\ast}(x)|_{H}
\]
and
\[
a(x)=Pb(x)|_{H}
\]
respectively, for all $x\in h$. These are bounded operators on $H$ which are
adjoints of one another, and satisfy the anti-commutation relations
\[
\{a(x),a(y)\}=0
\]
and
\[
\{a(x),a^{\ast}(y)\}=\left\langle x,y\right\rangle 1_{H}
\]
for all $x,y\in h$, where $1_{H}$ is the identity operator on $H$.

Next we introduce the lattice $L$. Let $L$ be any countable or finite set. We
assume that $L$ indexes an orthonormal basis for $h$, namely
\[
e_{l}
\]
for $l\in L$. We could, if necessary, rather denote the orthonormal basis as
$e_{l,s}$ where $s\in S$ specifies some further properties beyond the lattice
we are considering, say spin values. For simplicity of notation, however, we
subsume all such properties into the set $L$.

Also then write
\begin{equation}
f_{\varnothing}:=\Psi\label{fleeg}
\end{equation}
and
\begin{equation}
f_{(l_{1},...,l_{n})}:=Pe_{l_{1}}\otimes...\otimes e_{l_{n}} \label{f}
\end{equation}
for all $l_{1},...,l_{n}\in L$, for any $n=1,2,3,...$.

Given any subset $I$ of $L$, denote by
\[
D_{I}
\]
a set of finite sequences $(l_{1},...,l_{n})$ in $L$, for $n=0,1,2,3,...$,
with $l_{j}\neq l_{k}$ when $j\neq k$, such that each finite subset of $I$
corresponds to exactly one element of $D_{I}$. The empty subset of $I$
corresponds to the empty sequence denoted by $\varnothing\in D_{I}$, which is
the case $n=0$. Note that the vectors
\[
f_{M}
\]
with $M\in D_{L}$ form an orthonormal basis for $H$. The set $D_{I}$ is not
uniquely specified, but that does not matter, as it is just a way to label a
set of orthonormal vectors up to a factor $\pm1$ for each vector, irrespective
of how $D_{I}$ is chosen for a given $I$.

We also use the notation
\[
a_{l}^{\ast}=a^{\ast}(e_{l})
\]
and
\[
a_{l}=a(e_{l})
\]
for all $l\in L$. We note that
\[
a_{l}^{\ast}f_{(l_{1},...,l_{n})}=f_{(l,l_{1},...,l_{n})}
\]
(which is $0$ if $l\in\{l_{1},...,l_{n}\}$) and
\[
a_{l}f_{(l_{1},...,l_{n})}=(-1)^{k-1}f_{(l_{1},...,\hat{l}_{k},...,l_{n})}
\]
for $k$ such that $l_{k}=l$, where $(l_{1},...,\hat{l}_{k},...,l_{n})$ refers
to the sequence $(l_{1},...,l_{n})$ with $l_{k}$ removed, while $a_{l}
f_{(l_{1},...,l_{n})}=0$ if $l\notin\{l_{1},...,l_{n}\}$. These facts are
useful to keep in mind when manipulating expressions involving $f_{(l_{1}
,...,l_{n})}$.

For any subset $I$ of $L$, let
\[
A(I)
\]
denote the C*-subalgebra of $B(H)$ generated by $\{a_{l}:l\in I\}$, where
$B(H)$ is the C*-algebra of all bounded operators on $H$. Of course, since
$a_{l}^{\ast}$ is the adjoint of $a_{l}$, we have $a_{l}^{\ast}\in A(I)$ for
all $l\in L$. Because $\{a_{l},a_{l}^{\ast}\}=1_{H}$, it follows that $A(I)$
contains the unit $1_{H}$ of $B(H)$.

From the next section onwards we are going to focus on the case where $I$ is
finite, and $A(I)$ therefore finite dimensional because of the
anti-commutation relations. In this case we can simply view $A(I)$ as the
algebra generated by operators $a_{l}$ and $a_{l}^{\ast}$ for $l\in I$, and
C*-algebraic notions become less important. However, we nevertheless continue
with the usual C*-algebraic notation of denoting the adjoint of an operator
$a$ by $a^{\ast}$. Note that we need not assume that $L$ is finite.

\section{A fermionic entangled state}\label{AfdProdDiag}


Here our main goal is to construct the entangled state which is to play a
central role in the next section where detailed balance is discussed. However,
we first construct a fermionic analogue of a product of two states, each of
which is given by a diagonal density matrix in terms of basis vectors obtained
from the creation operators. This product state will not be used in fermionic
detailed balance, but gives some insight into how states for combined
fermionic systems should be constructed, which is instructive for the
subsequent construction of the entangled state. It may  also be
relevant in constructing examples of states for the more general case of
balance instead of detailed balance (see \cite{DS2}), but that will not be
treated in this paper. Our treatment of entanglement here is self-contained,
but brief and limited to the specific entangled state that we need, presented
in a mathematical form convenient for our later work in the paper. For more
systematic investigations into entanglement for systems of identical particles,
the reader is referred to \cite{ESBL} and \cite{GMW}. 
Also see \cite{GM, L, Mor, PY, S, WV, Z, ZW} for some of the other early papers on this 
topic, as well as \cite{BF, BFM, BFM2, BFM3, BFT, MB} for a selection of more recent ones.

Consider $A(I)$ and $A(J)$, where $I$ and $J$ are disjoint finite subsets of
$L$. The finite dimensionality of the algebras $A(I)$ and $A(J)$ due to $I$
and $J$ being finite, allows us to avoid any technicalities involving limits
and infinite sums. It should be possible to handle such technicalities, but
that will not help to clarify the conceptual aspects we want to focus on.

The fermion lattice now provides a convenient framework to construct fermionic
analogues of product or entangled states for $A(I)$ and $A(J)$.

Consider two sets of probabilities, $p_{M}$ for $M\in D_I$, and $q_{N}$ 
for $N\in D_J$, i.e. $p_{M}\geq0$,
$q_{N}\geq0$, $\sum_{M\in D_{I}}p_{M}=1$ and $\sum_{N\in D_{J}}q_{N}=1$. For
$A(I)$ and $A(J)$, we respectively consider the diagonal density matrices
\begin{equation}
\rho_{I}=\sum_{M\in D_{I}}p_{M}f_{M}\Join f_{M} \label{roI}
\end{equation}
and
\[
\rho_{J}=\sum_{N\in D_{J}}q_{N}f_{N}\Join f_{N}
\]
with
\[
x\Join y\in B(H)
\]
defined as
\[
(x\Join y)z=x\left\langle y,z\right\rangle
\]
for all $x,y,z\in H$, inspired by Dirac notation $\left|  x\right\rangle
\left\langle y\right|  $. I.e., $f_{M}\Join f_{M}$ could also be written as
$\left|  f_{M}\right\rangle \left\langle f_{M}\right|  $.

We aim to define a fermionic analogue of a product state for the two states
$\rho_{I}$ and $\rho_{J}$ such that the state $\rho$ which is obtained is
itself a sensible fermionic state. We achieve this by setting
\[
\rho=\sum_{M\in D_{I}}\sum_{N\in D_{J}}p_{M}q_{N}f_{MN}\Join f_{MN}
\]
where
\[
MN
\]
denotes the concatenation of the sequences $M$ and $N$, i.e. if $M$ and $N$
are the sequences $(m_{1},...,m_{j})$ and $(n_{1},...,n_{k})$ respectively,
then $MN$ denotes the sequence $(m_{1},...,m_{j},n_{1},...,n_{k})$, while for
$M=\varnothing$ we have $MN=N$, and for $N=\varnothing$ we have $MN=M$. The
mixed state $\rho$ is a fermionic state simply because the pure states
$f_{MN}$ are. Note that $\rho$ was constructed in analogy to the usual product
state given by
\[
\sum_{M\in D_{I}}\sum_{N\in D_{J}}
p_{M}q_{N}(f_{M}\otimes f_{N})\Join(f_{M}\otimes f_{N}),
\]
which is not a fermionic state in general, since $f_{M}\otimes f_{N}$ is not.

It is worth emphasizing (see \cite{GMW} and \cite{ESBL}) that despite the form
of the vectors $f_{MN}$, which is given by (\ref{f}) and (\ref{P}), they are
not viewed as being entangled states. Therefore the state $\rho$, being a
mixture of these pure states, possesses no entanglement. Similarly $\rho_{I}$
and $\rho_{J}$ possess no entanglement.

Next we consider the fermionic entangled state of main interest to us in this
paper. We assume that $J$ above has the same number of elements as $I$, but
still with $I\cap J=\varnothing$, and let
\[
\iota:I\rightarrow J
\]
be a bijection. The role of $\iota$ is to view $J$ as a ``copy'' of $I$
elsewhere in $L$, with the goal of constructing an entangled state of two
copies of the same state. We use $\rho_{I}$ as above, but replace $\rho_{J}$
by
\[
\rho_{\iota(I)}=\sum_{M\in D_{I}}p_{M}f_{\iota(M)}\Join f_{\iota(M)}
\]
where $\iota(M):=(\iota(m_{1}),...,\iota(m_{j}))$ for $M=(m_{1},...,m_{j})$.
The fermionic entangled state of interest to us is then defined to be
\begin{equation}
\Phi=\sum_{M\in D_{I}}p_{M}^{1/2}f_{M\iota(M)}\in H \label{diag}
\end{equation}
where $M\iota(M)$ again denotes concatenation as above. This is in analogy to
the entangled state
\begin{equation}
\sum_{M\in D_{I}}p_{M}^{1/2}f_{M}\otimes f_{\iota(M)} \label{diag2}
\end{equation}
which however is not a fermionic state, i.e. it is not in $H$, as
$f_{M}\otimes f_{\iota(M)}\notin H$.

\begin{remark}
The term ``diagonal state'' could also be used for $\Phi$, partly due to only
the form $f_{M\iota(M)}$ appearing, instead of the more general case
$f_{M\iota(N)}$, but also partly because in classical probability an analogous
construction leads to a so-called diagonal measure (in which entanglement
plays no role). This classical construction has a general noncommutative
counterpart (see \cite{D} and \cite{Fid}, as well as \cite[Subsection
7.2]{DS2}) which generalizes (\ref{diag2}), but not (\ref{diag}).
\end{remark}

We note that both $\rho$ and $\Phi$ reduce to the correct states, i.e.
\[
\operatorname*{Tr}(\rho a)=\operatorname*{Tr}(\rho_{I}a)
\]
and
\[
\operatorname*{Tr}(\rho b)=\operatorname*{Tr}(\rho_{J}b)
\]
for all $a\in A(I)$ and $b\in A(J)$, so $\rho$ reduces to $\rho_{I}$ and
$\rho_{J}$, while
\[
\left\langle \Phi,a\Phi\right\rangle =\operatorname*{Tr}(\rho_{I}a)
\]
and
\[
\left\langle \Phi,b\Phi\right\rangle =\operatorname*{Tr}(\rho_{\iota(I)}b)
\]
for all $a\in A(I)$ and $b\in A(\iota(I))$. This can be verified by fairly
straightforward calculations.

Note in particular that the pure state $\Phi$ reduces to the mixed states
$\rho_{I}$ and $\rho_{\iota(I)}$ for the algebras $A(I)$ and $A(\iota(I))$
respectively, each state being a mixture of fermionic pure states, confirming
that $\Phi$ is entangled if at least two of the probabilities $p_{M}$ are not zero.

To conclude, $\Phi$ is the state that will be of central importance in the
rest of the paper.

\section{Fermionic standard quantum detailed balance}\label{AfdFb}

In this section consider a purely fermionic formulation for detailed balance,
in terms of the framework set up so far, in particular making use of the state
$\Phi$ defined in (\ref{diag}).

We start by briefly reviewing the detailed balance condition in the generic
from not specifically involving indistinguishable particles. In particular we focus on
a standard quantum detailed balance of with respect to a reversing operation,
as defined in \cite{FR} and \cite{FU}, and also studied in \cite{BQ2} and
\cite{FR2}. We only discuss it for finite dimensional systems. The development
of quantum detailed balance more generally can be retraced in \cite{Ag, Al,
CWA, KFGV, M83, M, M84, MS, AI}.

Consider a quantum system with $n$ dimensional Hilbert space and its
observable algebra representable as the algebra $M_{n}$ of $n\times n$
matrices over $\mathbb{C}$. Let the system's state be given by the density
matrix $\rho$, and we then choose to work in an orthonormal basis
$d_{1},...,d_{n}$ in which this density matrix is diagonal, say
\[
\rho=\left[
\begin{array}
[c]{ccc}
p_{1} &  & \\
& \ddots & \\
&  & p_{n}
\end{array}
\right]  .
\]
Let
\[
\tau_{t}:M_{n}\rightarrow M_{n}
\]
be a semigroup of (completely) positive unital maps giving the dynamics of the
system as a function of time $t\geq0$. One then considers the following
entangled state of two copies of $(M_{n},\rho)$:
\[
\Omega=\sum_{j=1}^{n}p_{j}^{1/2}d_{j}\otimes d_{j}
\]
which we represent as a state $\omega$ on the composite system's observable
algebra $M_{n}\otimes M_{n}$ by
\begin{equation}
\omega(a)=\left\langle \Omega,a\Omega\right\rangle \label{omega}
\end{equation}
for all $a\in M_{n}\otimes M_{n}$.

In this set-up we can express standard quantum detailed balance of the system
with respect to a reversing operation, as the condition
\begin{equation}
\omega(a\otimes\tau_{t}(b))=\omega(\tau_{t}(a)\otimes b) \label{fb}
\end{equation}
for all $a,b\in M_{n}$ and all $t\geq0$. Here we have in effect made a
standard choice of reversing operation as the transposition of matrices with
respect to the chosen basis. In this paper we refer to condition (\ref{fb})
simply as \emph{standard quantum detailed balance} of the system 
$(M_{n},\tau,\rho)$. See \cite[Section 5]{DS} for more detail on this specific
formulation of the standard quantum detailed balance condition.

With this background in hand, we can now work towards writing down a fermionic
version of condition (\ref{fb}).

Returning to our notation from Section \ref{AfdProdDiag}, we are going to work
with the algebra $A(I\cup\iota(I))$ in the place of $M_{n}\otimes M_{n}$, and
in analogy with $\omega$ we define the state $\varphi$ on $A(I\cup\iota(I))$
by
\begin{equation}
\varphi(a)=\left\langle \Phi,a\Phi\right\rangle \label{fi}
\end{equation}
for all $a\in A(I\cup\iota(I))$. Consider a semigroup $\tau$ of positive (or
completely positive) unital maps
\[
\tau_{t}:A(I)\rightarrow A(I)
\]
for $t\geq0$, which is taken to be the dynamics of the system on $A(I)$, we
need this dynamics to be carried over to $A(\iota(I))$ in order to have a copy
of the dynamics $\tau$ on $A(\iota(I))$. So consider the $\ast$-isomorphism
\[
\eta:A(I)\rightarrow A(\iota(I))
\]
given by $\eta(a_{l})=a_{\iota(l)}$ for all $l\in I$. Then copy the dynamics
on $A(I)$ to $A(\iota(I))$ by
\[
\tau_{t}^{\iota}:A(\iota(I))\rightarrow A(\iota(I)):
a\mapsto\eta\circ\tau_{t}\circ\eta^{-1}.
\]
Similarly we can define $\alpha^\iota=\eta\circ\alpha\circ\eta^{-1}$ for any linear
 $\alpha:A(I)\rightarrow A(I)$.

In analogy to Eq. (\ref{fb}), we then say that $(A(I),\rho_{I},\tau)$
satisfies \emph{fermionic standard quantum detailed balance} when
\begin{equation}
\varphi(a\tau_{t}^{\iota}(b))=\varphi(\tau_{t}(a)b) \label{ffbKont}
\end{equation}
for all $a\in A(I)$ and $b\in A(\iota(I))$, and all $t\geq0$.

\begin{remark}
	Typically one would be interested in the case where all the probabilities
	 $p_M$ appearing in $\rho_I$, as given by (\ref{roI}), are non-zero. 
	 However, this is not 	 mathematically essential at the moment. In fact, 
	 it only becomes important when studying duality in Section \ref{AfdDuaal}.
\end{remark}

Note that the time variable $t$ does not play an essential role in what we
have done so far. We could equally well only consider a single (completely)
positive unital map $\tau:A(I)\rightarrow A(I)$, which means that we in effect
only consider one instant in time (or a discrete set of instants in time upon
iterating the single map). Then fermionic standard quantum detailed balance
of $(A(I),\rho_{I},\tau)$ is expressed as
\begin{equation}
\varphi(a\tau^{\iota}(b))=\varphi(\tau(a)b) \label{ffb}
\end{equation}
for all $a\in A(I)$ and $b\in A(\iota(I))$.

\section{An example}\label{AfdVb}


We exhibit a simple example of fermionic standard quantum detailed balance as
defined by Eq. (\ref{ffb}). The example is based on an example of the type
discussed in \cite[Section 6]{AFQ}, \cite{BQ}, \cite[Section 5]{FR10},
\cite[Subsection 7.1]{FR} and \cite[Section 7]{DS2}, but adapted to the
fermionic framework. At its core it can be viewed as consisting of
``balanced'' cycles. We concentrate on discrete time, but at the end of 
the section we explain how to extend this example to continuous time.

Start by defining a unitary operator
\[
U:H\rightarrow H
\]
by setting
\[
Uf_{\varnothing}=f_{\varnothing}
\]
and
\[
Uf_{(l_{1},...,l_{n})}=f_{(\sigma(l_{1}),...,\sigma(l_{n}))}
\]
for all finite sequences $(l_{1},...,l_{n})$ in $L$, for any $n=1,2,3,...$,
where
\[
\sigma:L\rightarrow L
\]
is a permutation of $I$. By this we mean that $\sigma|_{I}:I\rightarrow I$ is
a bijection, while $\sigma|_{L\backslash I}$ is the identity map on the
complement $L\backslash I$ of $I$. The operator $U$ is well-defined, since the
vectors $f_{(l_{1},...,l_{n})}$ include an orthonormal basis for $H$, as
mentioned in Section \ref{AfdProdDiag}, while it is easily checked from the
definition of $f_{(l_{1},...,l_{n})}$ in Eq. (\ref{f}) and Eq. (\ref{P}) that
any permutation of $l_{1},...,l_{n}$ is consistent with the definition of $U$.

It follows that
\[
U^{\ast}a_{l}^{\ast}Uf_{(l_{1},...,l_{n})}=U^{\ast}f_{(l,\sigma(l_{1}
),...,\sigma(l_{n}))}=f_{(\sigma^{-1}(l),l_{1},...,l_{n})}=a_{\sigma^{-1}
(l)}^{\ast}f_{(l_{1},...,l_{n})}
\]
so
\[
U^{\ast}a_{l}^{\ast}U=a_{\sigma^{-1}(l)}^{\ast}
\]
which means in particular that
\[
U^{\ast}A(I)U=A(I).
\]
Furthermore,
\begin{equation}
Ua_{l}^{\ast}U^{\ast}=a_{\sigma(l)}^{\ast} \label{UU*}
\end{equation}
so for any fixed $\lambda\in\lbrack0,1]$ we obtain a well-defined unital completely
positive map
\[
\tau:A(I)\rightarrow A(I)
\]
by setting
\begin{equation}
\tau(a)=\lambda U^{\ast}aU+(1-\lambda)UaU^{\ast} \label{tauVb}
\end{equation}
for all $a\in A(I)$. Keep in mind that $\sigma$ can be decomposed into cycles,
so in effect $\tau$ is built from two sets of cycles, the one set being
opposite to the other. It will shortly become clear that if we ``balance''
these opposite cycles by taking $\lambda=1/2$, then fermionic standard quantum detailed balance
emerges, analogous to the usual (or generic) case.

As in Section \ref{AfdProdDiag} we consider a bijection $\iota:I\rightarrow J$
where $I$ and $J$ are disjoint subsets of $L$, and as in Section \ref{AfdFb}
we copy the dynamics $\tau$ to $A(\iota(I))$. Explicitly we can do it as follows:

Define a permutation $\sigma^{\iota}:L\rightarrow L$ of $J$ by
\[
\sigma^{\iota}(l)=\iota\circ\sigma\circ\iota^{-1}(l)
\]
for $l\in J$, and
\[
\sigma^{\iota}(l)=l
\]
for $l\in L\backslash J$. Using this, we define a unitary operator
\[
V:H\rightarrow H
\]
by
\[
Vf_{\varnothing}=f_{\varnothing}
\]
and
\[
Vf_{(l_{1},...,l_{n})}=f_{(\sigma^{\iota}(l_{1}),...,\sigma^{\iota}(l_{n}))}
\]
for all finite sequences $(l_{1},...,l_{n})$ in $L$. Then we can define the
copy $\tau^{\iota}$ of $\tau$ on $A(\iota(I))$ by
\begin{equation}
\tau^{\iota}(b)=\lambda V^{\ast}bV+(1-\lambda)VbV^{\ast} \label{tauVbKopie}
\end{equation}
for all $b\in A(\iota(I))$.

Now, as opposed to the usual case of standard quantum detailed balance in Eq. (\ref{fb}), where
one uses the tensor product, we now make use of the properties of our
fermionic lattice, in particular the fact that $I$ and $\iota(I)$ are
disjoint, to show how fermionic standard quantum detailed balance is obtained.
The key point in this respect, is that from Eq. (\ref{UU*}) we deduce
\[
a_{l}^{\ast}U=Ua_{l}^{\ast}
\]
for all $l\in L\backslash I$, hence
\[
bU=Ub
\]
for all $b\in A(\iota(I))$, and similarly
\[
aV=Va
\]
for all $a\in A(I)$.

Using this we can do the following calculation for any $a\in A(I)$ and $b\in
A(\iota(I))$, in terms of the state $\varphi$ given by (\ref{fi}) and
(\ref{diag}), to obtain conditions under which fermionic standard quantum
detailed balance is satisfied:

Firstly,
\[
\varphi(\tau(a)b)=\lambda\left\langle \Phi,U^{\ast}aUb\Phi\right\rangle
+(1-\lambda)\left\langle \Phi,UaU^{\ast}b\Phi\right\rangle .
\]
But, using the notation $\sigma(M):=(\sigma(l_{1}),...,\sigma(l_{n}))$, when
$M=(l_{1},...,l_{n})$, we obtain
\begin{align*}
&  \left\langle \Phi,U^{\ast}aUb\Phi\right\rangle \\
&  =\left\langle U\Phi,abU\Phi\right\rangle \\
&  =\left\langle \sum_{M\in D_{I}}p_{M}^{1/2}f_{\sigma(M)\iota(M)}
,ab\sum_{N\in D_{I}}p_{N}^{1/2}f_{\sigma(N)\iota(N)}\right\rangle \\
&  =\left\langle \sum_{\{M:\sigma^{-1}(M)\in D_{I}\}}p_{\sigma^{-1}(M)}
^{1/2}f_{M\iota(\sigma^{-1}(M))},ab\sum_{\{N:\sigma^{-1}(N)\in D_{I}
\}}p_{\sigma^{-1}(N)}^{1/2}f_{N\iota(\sigma^{-1}(N))}\right\rangle \\
&  =\left\langle \sum_{\{M:\sigma^{-1}(M)\in D_{I}\}}p_{\sigma^{-1}(M)}
^{1/2}V^{\ast}f_{M\iota(M)},ab\sum_{\{N:\sigma^{-1}(N)\in D_{I}\}}
p_{\sigma^{-1}(N)}^{1/2}V^{\ast}f_{N\iota(N)}\right\rangle \\
&  =\left\langle \Phi_{1},aVbV^{\ast}\Phi_{1}\right\rangle
\end{align*}
where
\[
\Phi_{1}:=\sum_{\{M:\sigma^{-1}(M)\in D_{I}\}}p_{\sigma^{-1}(M)}
^{1/2}f_{M\iota(M)}=\sum_{M\in D_{I}}p_{M}^{1/2}f_{\sigma(M)\iota(\sigma
(M))}.
\]
Now, for any sequence $M\in D_{I}$, let
\[
\sigma_{M}:L\rightarrow L
\]
be the permutation of $M$ such that
\[
\sigma\circ\sigma_{M}(M)\in D_{I}.
\]
The point here is that $\sigma(M)$ is a sequence in $I$, but it need not be in
the correct order to be an element of $D_{I}$. The permutation $\sigma_{M}$
corrects for this. If we assume that the density matrix $\rho_{I}$ in
(\ref{roI}) satisfies
\begin{equation}
p_{\sigma_{M}^{-1}\circ\sigma^{-1}(M)}=p_{M} \label{inv}
\end{equation}
for all $M\in D_{I}$, then it follows from the definition of 
$f_{(l_{1},...,l_{n})}$ in Eqs. (\ref{P}), (\ref{fleeg}) and (\ref{f}), 
in particular how the sign of $f_{(l_{1},...,l_{n})}$ may change due to 
permutation, that
\[
\Phi_{1}=\sum_{M\in D_{I}}p_{M}^{1/2}f_{\sigma\circ\sigma_{M}(M)\iota
(\sigma\circ\sigma_{M}(M))}=\sum_{M\in D_{I}}p_{\sigma_{M}^{-1}\circ
\sigma^{-1}(M)}^{1/2}f_{M\iota(M)}=\Phi.
\]
A simpler (but less general) assumption that ensures Eq. (\ref{inv}), is
\begin{equation}
p_{M}=p_{N} \label{inv2}
\end{equation}
for any pair of sequences $M,N\in D_{I}$ that have the same length.

Then
\[
\left\langle \Phi,U^{\ast}aUb\Phi\right\rangle =\left\langle \Phi,aVbV^{\ast
}\Phi\right\rangle .
\]

If we let
\[
\bar{\sigma}_{M}:L\rightarrow L
\]
be the permutation of $M\in D_{I}$ such that
\[
\sigma^{-1}\circ\bar{\sigma}_{M}(M)\in D_{I},
\]
and we assume that the density matrix $\rho_{I}$ satisfies
\begin{equation}
p_{\bar{\sigma}_{M}^{-1}\circ\sigma(M)}=p_{M} \label{inv'}
\end{equation}
for all $M\in D_{I}$, we also obtain
\[
\left\langle \Phi,UaU^{\ast}b\Phi\right\rangle =\left\langle \Phi,aV^{\ast
}bV\Phi\right\rangle .
\]
From this we conclude that if conditions (\ref{inv}) and (\ref{inv'}) are
satisfied, or alternatively just condition (\ref{inv2}), then
\[
\varphi(\tau(a)b)=\lambda\left\langle \Phi,aVbV^{\ast}\Phi\right\rangle
+(1-\lambda)\left\langle \Phi,aV^{\ast}bV\Phi\right\rangle .
\]
Under these conditions, if
\[
\lambda=1/2,
\]
then
\[
\varphi(\tau(a)b)=\varphi(a\tau^{\iota}(b))
\]
for all $a\in A(I)$ and $b\in A(\iota(I))$. That is, $(A(I),\rho_{I},\tau)$ then satisfies fermionic standard quantum detailed balance.

This example can easily be adapted to a continuous time example by using the
form $\mathcal{L}(a)=\lambda U^{\ast}aU+(1-\lambda)UaU^{\ast}-a$ as the
generator of a quantum Markov semigroup which then satisfies fermionic
detailed balance in the form (\ref{ffbKont}).

\section{Detailed balance vs fermionic detailed balance}\label{AfdTeenVb}


Here we show that the standard quantum detailed balance and fermionic standard
quantum detailed balance conditions discussed in Section \ref{AfdFb}, are not
equivalent when we consider dynamics on an $A(I)$. More
precisely, we exhibit a simple example of a unital completely positive map on
$A(I)$ satisfying standard quantum detailed balance but not fermionic standard
quantum detailed balance. As is the case with the dynamics in Section
\ref{AfdFb}, this example will nevertheless still have even dynamics, the
definition of which is also discussed in this section.

Consider the following example of standard quantum detailed balance, which is
similar to the example discussed in Section \ref{AfdVb}, but not tailored to
the fermionic case:

Consider a $n$-dimensional Hilbert space $K$ with orthonormal basis
$d_{1},...,d_{n}$. Given a permutation $\varpi\in S_{n}$ of $\{1,...,n\}$, we
define a unitary operator $W:K\rightarrow K$ by
\[
Wd_{j}=d_{\varpi(j)}
\]
for $j=1,...,n$. That is, we are considering a permutation of the orthonormal
basis. For $0\leq\lambda\leq1$ we then consider the unital completely positive
map $\alpha:B(K)\rightarrow B(K)$ given by
\[
\alpha(a)=\lambda W^{\ast}aW+(1-\lambda)WaW^{\ast}
\]
for all $a\in B(K)$. Now, using the basis $d_{1},...,d_{n}$ to represent
$B(K)$ as $M_{n}$ as in Section \ref{AfdFb}, assume that the probabilities
$p_{j}$ from Section \ref{AfdFb} are equal for the basis vectors $d_{j}$ lying
in the same cycle of the decomposition of $\sigma$ into cycles. It is then
straightforward to check that standard quantum detailed balance is satisfied
when 
$\lambda=1/2$
(the argument is similar, but notationally simpler than that in Section
\ref{AfdVb}).

In particular, keeping in mind that $A(I)$ is isomorphic to $M_{2^{|I|}}$,
we can in this way obtain dynamics on $A(I)$ satisfying standard quantum
detailed balance, using the vectors $f_{M}$, $M\in D_{I}$, as the orthonormal
basis in the place of $d_{1},...,d_{n}$, and applying a permutation to this
basis. In particular, we see from this argument that the dynamics on $A(I)$ in
Section \ref{AfdVb} satisfies (usual) standard quantum detailed balance when
$\lambda=1/2$, if the probabilities $p_{M}$ are equal for basis vectors
$f_{M}$ lying in the same cycle of the permutation. However, if the
permutation of the orthonormal basis, namely the vectors $f_{M}$, is not
obtained from a permutation of $I$ as in Section \ref{AfdVb}, then the
argument in Section \ref{AfdVb} falls apart. In that case we can not expect to
have fermionic standard quantum detailed balance in general, even though
standard quantum detailed balance is satisfied as explained above.

We can illustrate this explicitly, while preserving a basic property that the
dynamics $\tau$ in Section \ref{AfdVb} has, namely that it is even. We
describe this concept before continuing to our example:


Define a unitary operator
\[
\theta:H\rightarrow H
\]
on the Fermi-Fock space via $\theta f_{M}=f_{M}$ if the sequence $M$ has even
length, while $\theta f_{M}=-f_{M}$ if the sequence $M$ has odd length. Note
that $\theta=\theta^{\ast}$. Furthermore, it is easily confirmed that
\[
\theta a_{l}^{\ast}\theta=-a_{l}^{\ast}
\]
for all $l\in L$, by applying $\theta a_{l}^{\ast}\theta$ to the basis vectors
$f_{M}$. By taking the adjoint both sides, we also have
\[
\theta a_{l}^{\ast}\theta=-a_{l}^{\ast}
\]
for all $l\in L$. Therefore we can define a $\ast$-automorphism 
$\Theta_{I}$ of $A(I)$ by
\[
\Theta_{I}(a)=\theta a\theta
\]
for all $a\in A(I)$. This works for every subset $I$ of $L$, including $I=L$,
so we may as well just work with $\Theta=\Theta_{L}$, since then $\Theta_{I}$
is just the restriction of $\Theta$ to $A(I)$. (Also see \cite[Section 4.1]{AM}.)

One can then show that the dynamics $\tau$ in Section \ref{AfdVb} is
\emph{even} for all $\lambda\in\lbrack0,1]$, by which we mean that
\[
\tau\circ\Theta=\Theta\circ\tau.
\]
This follows from $U\theta=\theta U$, which is true, since in $Uf_{(l_{1}
	,...,l_{n})}=f_{(\sigma(l_{1}),...,\sigma(l_{n}))}$ the sequences
$(l_{1},...,l_{n})$ and $(\sigma(l_{1}),...,\sigma(l_{n}))$ have the same
length, so both are even or both are odd.

We now construct an example of dynamics which is even and satisfies standard
quantum detailed balance, but not fermionic standard quantum detailed balance:


Consider the case $|I|=2$, i.e. $A(I)$ is generated by two operators $a_{l}$;
let us call them $a_{1}$ and $a_{2}$. We use $D_{I}=\{\varnothing,(1),(2),(1,2)\}$. Furthermore, for the rest of this section, we set
\[
p_{\varnothing}=p_{(1)}=p_{(2)}=p_{(1,2)}=\frac{1}{4}
\]
as the probabilities appearing in $\rho_{I}$ given by (\ref{roI}), 
and in terms of which we express the detailed balance conditions. 
We take $|J|=2$ as well, with $I\cap J$, and let $a_{3}$ and $a_{4}$ 
denote the generators of $A(J)$. The bijection
$\iota:I\rightarrow J$ we use is given by
\[
\iota(1)=3\text{, \ }\iota(2)=4.
\]

The dynamics on $A(I)$ will be obtained from a unitary operator
\[
U_{I}:H_I\rightarrow H_I,
\]
where $H_{I}$ is the subspace of $H$ spanned by the set 
$\{f_{\varnothing},f_{(1)},f_{(2)},f_{(1,2)}\}$, and $U_I$ is defined via
\[
U_{I}f_{\varnothing}=f_{(1)}\text{, \ }U_{I}f_{(1)}=f_{(1,2)}\text{, \ }
U_{I}f_{(1,2)}=f_{(2)}\text{, \ }U_{I}f_{(2)}=f_{\varnothing}.
\]
Note that this is a permutation of the basis not given by a permutation of the
set $I=\{1,2\}$.

View $A(I)$ as being faithfully represented on $H_{I}$ by $\pi_{I}$ via
\[
\pi_{I}(a):=a|_{H_{I}},
\]
i.e. we restrict the elements of $A(I)$ to $H_{I}$. In other words, we
represent $A(I)$ faithfully as $\pi_{I}(A(I))$, which is isomorphic to $M_{4}$. 
The reason for this is that we have not defined $U_{I}$ on the whole of
$H$; trying to extend $U_{I}$ to $H$, and working with $A(I)$ itself, is
inconvenient in this case. However, for simplicity we suppress the $\pi_I$ in
our notation below. Analogously for $A(J)$ on $H_{J}$.

By applying $U_{I}^{\ast}a_{1}U_{I}$ to the basis vectors $f_{\varnothing}$,
$f_{(1)}$, $f_{(2)}$ and $f_{(1,2)}$ of $H_{I}$, one can verify the formula
\[
U_{I}^{\ast}a_{1}U_{I}=a_{2}^{\ast}[a_{1},a_{1}^{\ast}]
\]
where $[\cdot,\cdot]$ denotes the commutator. Similarly we have the formulas
\[
U_{I}a_{1}U_{I}^{\ast}=a_{2}[a_{1},a_{1}^{\ast}],
\]
\[
U_{I}^{\ast}a_{2}U_{I}=a_{1}[a_{2}^{\ast},a_{2}],
\]
and
\[
U_{I}a_{2}U_{I}^{\ast}=a_{1}^{\ast}[a_{2},a_{2}^{\ast}].
\]

For any fixed $\lambda\in\lbrack0,1]$ we can then define dynamics $\alpha$ on $A(I)$
by
\[
\alpha(a)=\lambda U_{I}^{\ast}aU_{I}+(1-\lambda)U_{I}aU_{I}^{\ast}
\]
for all $a\in A(I)$. It is straightforward to check from the formulas above,
that $\alpha$ is even, i.e. $\alpha\circ\Theta=\Theta\circ\alpha$.

We copy $\alpha$ to $A(J)$ via
\[
\alpha^{\iota}(b)=\lambda V_{J}^{\ast}bV_{J}+(1-\lambda)V_{J}bV_{J}^{\ast},
\]
using the correspondingly defined unitary operator $V_{J}$ on $H_{J}$ given by
\[
V_{J}f_{\varnothing}=f_{(3)}\text{, \ }V_{J}f_{(3)}=f_{(3,4)}\text{, \ }
V_{J}f_{(3,4)}=f_{(4)}\text{, \ }V_{J}f_{(4)}=f_{\varnothing},
\]
and for which corresponding formulas as for $U_{I}$ above hold.

Let us now study detailed balance for the case $\lambda=1/2$. As already
explained earlier in this section, standard quantum detailed balance is then
satisfied. However, fermionic standard quantum detailed balance is not:

\[
\varphi(\alpha(a_{1})a_{4}^{\ast})=\frac{1}{4}
\]
while
\[
\varphi(a_{1}\alpha^{\iota}(a_{4}^{\ast}))=-\frac{1}{4},
\]
for $\varphi$ in (\ref{fi}). This can be verified using (\ref{diag}), which
here is
\[
\Phi=\frac{1}{2}(f_{\varnothing}+f_{(1,3)}+f_{(2,4)}+f_{(1,2,3,4)}),
\]
as well as
\[
\alpha(a_{1})=\frac{1}{2}(a_{2}+a_{2}^{\ast})[a_{1},a_{1}^{\ast}]
\]
and
\[
\alpha^{\iota}(a_{4})=\frac{1}{2}(a_{3}^{\ast}-a_{3})[a_{4},a_{4}^{\ast}],
\]
and by then calculating that
\[
\alpha(a_{1}^{\ast})\Phi=\alpha(a_{1})^{\ast}\Phi=
\frac{1}{4}(f_{(2)}+f_{(4)}+f_{(1,2,3)}+f_{(1,3,4)}),
\]
\[
a_{4}^{\ast}\Phi=\frac{1}{2}(f_{(4)}+f_{(1,3,4)}),
\]
\[
a_{1}^{\ast}\Phi=\frac{1}{2}(f_{(1)}+f_{(1,2,4)}),
\]
and
\[
\alpha^{\iota}(a_{4}^{\ast})\Phi=
\frac{1}{4}(-f_{(1)}-f_{(3)}-f_{(1,2,4)}-f_{(2,3,4)}),
\]
and then evaluating 
$\varphi(\alpha(a_{1})a_{4}^{\ast})=
\langle\alpha(a_{1}^{\ast})\Phi,a_{4}^{\ast}\Phi\rangle$ 
and 
$\varphi(a_{1}\alpha^{\iota}(a_{4}^{\ast}))=
\langle a_{1}^{\ast}\Phi,\alpha^{\iota}(a_{4}^{\ast})\Phi\rangle $.

Thus $(A(I),\rho_{I},\alpha)$ is indeed an example with 
even dynamics satisfying standard quantum detailed balance, 
but not fermionic standard quantum detailed balance.

\section{Duality}\label{AfdDuaal}


Certain types of duals (or adjoints) of dynamical maps play an important role
in quantum detailed balance conditions. See for example \cite{MS, FR, DS,
DS2}. Here we show that some, but not all, aspects of duality survive in our
framework for fermionic standard quantum detailed balance. Our discussion here
is of a preliminary nature and we suspect that it should be possible to develop
duality in the fermionic case further.

A very basic duality appearing in relation to standard quantum detailed
balance arises from the following bilinear form, defined in terms of the state
$\omega$ given in (\ref{omega}):
\[
B_{\omega}:M_{n}\times M_{n}\rightarrow\mathbb{C}:(a,b)\mapsto\omega(a\otimes
b).
\]
We note that for $a,b\geq0$, in the usual operator algebraic sense, i.e. $a$
and $b$ are self-adjoint operators with non-negative spectra, we have
\begin{equation}
B_{\omega}(a,b)\geq0. \label{pos}
\end{equation}
For a linear map $\alpha:M_{n}\rightarrow M_{n}$, the unique dual map
$\alpha^{\prime}:M_{n}\rightarrow M_{n}$ such that
\[
B_{\omega}(\alpha(a),b)=B_{\omega}(a,\alpha^{\prime}(b))
\]
for all $a,b\in M_{n}$, is then of some importance in connection to standard
quantum detailed balance and can for example be used to define the KMS-dual of
a positive map (again see the references mentioned above, for example
\cite[Definition 2.9]{DS2}, but also \cite{Pet} and \cite[Proposition
8.3]{OPet}). The positivity of $B_{\omega}$ mentioned above is
necessary in showing $\alpha^{\prime}$ is $n$-positive when $\alpha$ is. These
points in fact hold in a much more general infinite dimensional von Neumann
algebraic setup, not just on $M_{n}$ (see \cite[Proposition 3.1]{AC} and
\cite{DS2}). The special case $M_{n}$ fits into the general von Neumann
algebraic framework by representing the first copy of $M_{n}$ appearing in
$B_{\omega}$ as $M_{n}\otimes1_{n}$, with $1_{n}$ the $n\times n$ identity
matrix, while the commutant $1_{n}\otimes M_{n}$ of this representation serves
as the second copy of $M_{n}$ in $B_{\omega}$.

In analogy to this, we can study the following bilinear form in the fermionic
case, in terms of the state $\varphi$ given by (\ref{fi}), and where $I$ is
again a finite subset of $L$, and $\iota:I\rightarrow L$ is an injection such
that $I$ and $\iota(I)$ are disjoint:
\[
B_{\varphi}:A(I)\times A(\iota(I))\rightarrow\mathbb{C}:(a,b)\mapsto
\varphi(ab).
\]
Note that as in Section \ref{AfdFb}, the role of the tensor product structure
is here essentially taken over by the fermion lattice structure. Below it will
be seen that $B_{\varphi}$ does allow us to define dual maps similar to
$B_{\omega}$, but on the other hand it does not satisfy the positivity
property $B_{\varphi}(a,b)\geq0$ for all $a\geq0$ and $b\geq0$.

Note that in the definition of $B_{\varphi}$ we are in effect making a choice,
since we could equally well have used the definition 
$B_{\varphi}(a,b)=\varphi(ba)$. In $B_{\omega}$ no corresponding choice had to be made,
since $(a\otimes1)(1\otimes b)=a\otimes b=(1\otimes b)(a\otimes1)$. This is an
indication that the definition of $B_{\varphi}$ is not quite as natural as
that of $B_{\omega}$.

In order to show that we can define a dual in the fermionic case, analogous to
the usual case above, we first study the relevant properties of $B_{\varphi}$
as defined above. The main technical property we need is the following
non-degeneracy of $B_{\varphi}$:

\begin{proposition}
\label{nie-ont}In the definition of $\Phi$, given by (\ref{diag}), assume that
$p_{M}>0$ for all $M\in D_{I}$. Then:

(i) If $B_{\varphi}(a,b)=0$ for all $a\in A(I)$, for some $b\in A(\iota(I))$,
then $b=0$.

(ii) If $B_{\varphi}(a,b)=0$ for all $b\in A(\iota(I))$, for some $a\in A(I)$,
then $a=0$.
\end{proposition}

\begin{proof}
For any $M\in D_{I}$, let $a_{M}:=a_{m_{s}}...a_{m_{1}}$ and 
$a_{\iota(M)}:=a_{\iota(m_{s})}...a_{\iota(m_{1})}
$, where $M$ is the sequence
$(m_{1},...,m_{s})$, while we set $a_{M}=a_{\iota(M)}=1$ if $M$ is the empty
sequence. Furthermore, let $M^{\prime}\in D_{I}$ denote the complement of
$M\in D_{I}$, i.e. $M^{\prime}$ is the sequence in $D_{I}$ consisting of the
elements of $I$ not appearing in the sequence $M$. Keep in mind that the
algebra $A(\iota(I))$ is spanned by the linearly independent set of $2^{2|I|}$
elements given by 
\[
a_{\iota(N)}^{\ast}a_{\iota(M)}
\]
for $M,N\in D_{I}$, i.e. these elements form a basis for $A(\iota(I))$.

Note that for any given $M\in D_{I}$ and $b\in A(\iota(I))$ we then have from
$I\cap\iota(I)=\varnothing$ and 
$\{a_{k},a_{l}\}=\{a_{k},a_{l}^{\ast}\}=\{a_{k}^{\ast},a_{l}\}=
\{a_{k}^{\ast},a_{l}^{\ast}\}=0$ for $k\in I$ and
$l\in\iota(I)$, that
\[
a_{M^{\prime}}^{\ast}a_{M}b=ca_{M^{\prime}}^{\ast}a_{M}
\]
where $c$ is obtained from $b$ by expanding $b$ in the above mentioned basis
for $A(\iota(I))$, and changing the signs of some terms due to the
anti-commutation relations mentioned above. (Which terms change sign will not
be important to us.) So $c$ depends on $b$ and $|I|$, since 
$a_{M^{\prime}}^{\ast}a_{M}$ contains $|I|$ factors of $a_{k}$ or 
$a_{k}^{\ast}$ in total.

In addition we have
\[
a_{M^{\prime}}^{\ast}a_{M}f_{M\iota(M)}=f_{M^{\prime}\iota(M)}
\]
and
\[
a_{M^{\prime}}^{\ast}a_{M}f_{N\iota(N)}=0
\]
for sequences $N\in D_I$ such that $N\neq M$. Therefore, from (\ref{diag}),
\[
a_{M^{\prime}}^{\ast}a_{M}b\Phi=p_{M}^{1/2}cf_{M^{\prime}\iota(M)}.
\]
Similarly it now follows that
\[
a_{M^{\prime}}a_{M^{\prime}}^{\ast}a_{M}b\Phi=
p_{M}^{1/2}d_{|M|}f_{\iota(M)}.
\]
where $d_{|M|}$ is obtained from $c$ via 
$a_{M^{\prime}}c=d_{|M|}a_{M^{\prime}}$ 
by again changing the signs of certain terms. Note that in the process a
dependence on the length of the sequence $M^{\prime}$, so equivalently on that
of $M$, appears, which we indicate by $|M|$.

Assume that $B_{\varphi}
(a,b)=0$ for all $a\in A(I)$. Then it follows for all
$M,N\in D_{I}$ that
\begin{align*}
\left\langle f_{\iota(N)},d_{|M|}f_{\iota(M)}\right\rangle  
&  =\frac{1}{(p_{M}p_{N})^{1/2}}\left\langle a_{N'}a_{N'}^{\ast}a_{N}\Phi,
a_{M'}a_{M'}^{\ast}a_{M}b\Phi\right\rangle \\
&  =\frac{1}{(p_{M}p_{N})^{1/2}}
B_{\varphi}(a_{N}^{\ast}a_{N'}a_{N'}^{\ast}a_{M'}a_{M'}^{\ast}a_{M},b)\\
&  =0
\end{align*}

Even though the signs of the terms in $d_{|M|}$ may vary as $M$ varies, the
fact that
\[
\left\langle f_{\iota(N)},
a_{\iota(N)}^{\ast}a_{\iota(M)}f_{\iota(M)}\right\rangle 
=\langle f_{\varnothing},f_{\varnothing}\rangle
=1,
\]
whereas $a_{\iota(M_{1})}f_{\iota(M_{2})}=0$ if the sequence $M_{1}$ contains
entries not present in $M_{2}$, means that each term in $d_{|M|}$ is zero. To
see this, start with the shortest sequence $M=N=\varnothing$ in $D_{I}$, i.e.
start with 
$\left\langle f_{\varnothing},d_{|\varnothing|}f_{\varnothing}\right\rangle=0$, 
to see that the basis element
$a_{\iota(\varnothing)}^{\ast}a_{\iota(\varnothing)}=1$ 
in the expansion of $d_{|\varnothing|}$, has coefficient $0$. 
Then progressively check longer sequences $M$ and $N$ in 
$\left\langle f_{\iota(N)},d_{|M|}f_{\iota(M)}\right\rangle$,
step by step. Therefore all basis elements in the expansion of $d_{|M|}$ have 
zero coefficients. Likewise for $b$, as its coefficients are the same, 
except possibly for the signs, i.e. $b=0$, proving (i).

Similarly for (ii).
\end{proof}

From this proposition we obtain the following:

\begin{corollary}
\label{linFunk}(i) Every linear functional $f$ on $A(I)$ is of the form
\[
f=B_{\varphi}(\cdot,b)
\]
for some $b\in A(\iota(I))$ uniquely determined by $f$.

(ii) Every linear functional $g$ on $A(\iota(I))$ is of the form
\[
g=B_{\varphi}(a,\cdot)
\]
for some $a\in A(\iota(I))$ uniquely determined by $g$.
\end{corollary}

\begin{proof}

Consider the linear map
\[
F:A(\iota(I))\rightarrow A(I)^{\ast}
:b\mapsto B_{\varphi}(\cdot,b)
\]
where $A(I)^{\ast}$ denotes the dual of $A(I)$, i.e. the space of all linear
functionals on $A(I)$. In order to show (i), we only need to prove that $F$ is
a bijection.

Because of Proposition \ref{nie-ont}, $F$ is injective. It follows that
\begin{align*}
\dim A(I)^{\ast}
&  =\dim A(I)=\dim M_{2^{|I|}}
=\dim A(\iota(I))\\
&  =\dim\ker F+\dim F(A(\iota(I))=\dim F(A(\iota(I))
\end{align*}
which means that $F$ is also surjective, i.e. $F$ is bijective as needed.

Claim (ii) follows similarly.
\end{proof}

Now we are in a position to obtain the dual of a linear map in terms of
$B_{\varphi}$.

\begin{theorem}
(i) Any linear map
\[
\alpha:A(I)\rightarrow A(I).
\]
has a unique \emph{fermionic dual} map
\[
\alpha^{\varphi}:A(\iota(I))\rightarrow A(\iota(I))
\]
(which is necessarily linear) such that
\[
B_{\varphi}(\alpha(a),b)=B_{\varphi}(a,\alpha^{\varphi}(b))
\]
for all $a\in A(I)$ and $b\in A(\iota(I))$.

(ii) Any linear map
\[
\beta:A(\iota(I))\rightarrow A(\iota(I))
\]
has a unique \emph{fermionic dual} map
\[
\beta^{\varphi}:A(I)\rightarrow A(I)
\]
(which is necessarily linear) such that
\[
B_{\varphi}(a,\beta(b))=B_{\varphi}(\beta^{\varphi}(a),b)
\]
for all $a\in A(I)$ and $b\in A(\iota(I))$.

(iii) In terms of (i) and (ii),
\[
\alpha^{\varphi\varphi}=\alpha\text{ \ and \ }\beta^{\varphi\varphi}=\beta
\]
\end{theorem}

\begin{proof}
For $b\in A(\iota(I))$, define $f_{b}\in A(I)^{\ast}$ by 
$f_{b}(a):=B_{\varphi}(\alpha(a),b)$. By Corollary \ref{linFunk}
there is a unique element of
$A(\iota(I))$ which we denote by $\alpha^{\varphi}(b)$, such that
$f_{b}=B_{\varphi}(\cdot,\alpha^{\varphi}(b))$ proving (i). 
Part (ii) follows similarly. Part (iii) follows from uniqueness.
\end{proof}

\begin{corollary}
	When fermionic standard quantum detailed balance holds for
	$(A(I),\rho_{I},\tau)$ as in (\ref{ffbKont}), then the map 
	$\tau_{t}^{\iota}$ (for any $t$) is the unique map 	satisfying 
	(\ref{ffbKont}) for all $a\in A(I)$ and $b\in A(\iota(I))$. 
	Similarly for (\ref{ffb}).
\end{corollary}

Lastly we show by example that $B_{\varphi}$ does not satisfy a positivity
property of a form analogous to that satisfied by $B_{\omega}$ in (\ref{pos}).
This is an indication that the duality in the fermionic setup may not be as
useful as the duality obtained in the usual case in terms of $B_{\omega}$.

\begin{example}
Let $a:=(1+\kappa c)^{\ast}(1+\kappa c)\in A(I)$ and 
$b:=(1+\lambda d)^{\ast}(1+\lambda d)\in A(\iota(I))$, where $c=a_l$ 
for some $l\in I$, $d=a_{\iota(l)}$ and $\kappa,\lambda\in\mathbb{C}$. Then
\[
ab=q+(\kappa c+\bar{\kappa}c^{\ast})b
\]
where (due to the anti-commutation relations)
\[
q:=(1+\lambda d)^{\ast}(1+|\kappa|^{2}c^{\ast}c)(1+\lambda d)\geq0.
\]
Now, $\varphi(q)\geq0$ and one can show from (\ref{diag}) and (\ref{fi}) that
\[
r:=\varphi((\kappa c+\bar{\kappa}c^{\ast})b)=
\kappa\lambda\varphi(cd)+\bar{\kappa}\bar{\lambda}\varphi(c^{\ast}d^{\ast}).
\]
The complex conjugate of $\varphi(cd)$ is $\overline{\varphi(cd)}=\varphi((cd)^{\ast})=-\varphi(c^{\ast}d^{\ast})$, 
since $\{c^{\ast},d^{\ast}\}=0$. From the form of $\Phi$ in (\ref{diag}), one can see that $\varphi(cd)$ is real. So
$\varphi(cd)=-\varphi(c^{\ast}d^{\ast})$. 
Therefore
\[
r=(\kappa\lambda-\bar{\kappa}\bar{\lambda})\varphi(cd).
\]
It is simple to check that $\varphi(cd)$ can be non-zero for suitable 
choices of the probabilities $p_{M}>0$ in (\ref{diag}). Hence 
$r\notin\mathbb{R}$ is possible. Then
\[
B_{\varphi}(a,b)=\varphi(ab)\notin\mathbb{R}
\]
despite the fact that $a\geq0$ and $b\geq0$.
\end{example}

\section{Questions}\label{Vrae}


There are some natural further questions that could be explored:

Can the fermionic standard quantum detailed balance condition lead to more
refined results than the usual quantum detailed balance conditions when
applied to fermionic systems? 

What other interesting examples, aside from the one in Section \ref{AfdVb},
are there of fermionic standard quantum detailed balance?

Can duals of maps in the fermionic case be approached in a different way from
Section \ref{AfdDuaal}, to have better positivity properties?

What other forms of quantum detailed balance, aside from standard quantum
detailed balance with respect to a reversing operation, can similarly be
tailored to the fermionic case?

In this paper we have essentially just considered finite systems in a very
concrete set-up. What about more general fermionic systems and a more 
abstract set-up? At least an infinite version (i.e. infinite $I$) 
should be possible in the same concrete setting that we have used in this 
paper.

Lastly, can a bosonic version of detailed balance be developed in an analogous
way? This seems plausible, but may be technically more demanding.

\section*{Acknowledgements}

I thank Vito Crismale and Francesco Fidaleo for stimulating discussions at the
beginning of this project.

\end{document}